\documentclass[a4paper]{article}

\usepackage[english]{babel}
\usepackage[utf8]{inputenc}
\usepackage{amsmath}
\usepackage{amsthm}
\usepackage{mathrsfs}
\usepackage{mathtools}
\usepackage{amssymb}
\usepackage{graphicx}
\usepackage[colorinlistoftodos]{todonotes}
\usepackage{xcolor}
\usepackage{verbatim}
\usepackage{theoremref}

\numberwithin{equation}{section}
 
 \newcommand{\be}{\begin{equation}}
 \newcommand{\ee}{\end{equation}}
 \newcommand{\bb}[1]{\mathbb{#1}}
 \newcommand{\taubar}{{\bar \tau}}
 \newcommand{\pbar}{{\bar p}}
 \newcommand{\wtilde}{{\tilde w}}
 \newcommand{\Dtilde}{{\tilde D}}
 \newcommand{\Stilde}{{\tilde S}}
 \newcommand{\rtilde}{{\tilde r}}
 \newcommand{\ptilde}{{\tilde p}}
 \newcommand{\utilde}{{\tilde u}}
 
 \newcommand{\rv}{{\bf r}}
 \newcommand{\pv}{{\bf p}}
 \newcommand{\attesa}[1]{{\langle #1 \rangle_{\nu_{\tau,\pbar,\beta}}}}
 \newcommand{\somma}[3]{{\sum_{{#1}={#2}}^{#3}}}
 
 \newcommand{\PI}[1]{{p \left(t,\dfrac{#1}{N}\right)}}
 \newcommand{\RHO}[1]{{r \left(t,\dfrac{#1}{N}\right)}}
 \newcommand{\TAU}[1]{{\hat \tau \left(t,\frac{#1}{N}\right)}}

 \newcommand{\e}[1]{e^{\scalebox{0.7}{$\displaystyle #1$}}}
 
 \def\deltauno{\delta_1}
\def\deltadue{\delta_2}

\theoremstyle{plain}
\newtheorem{thm}{Theorem}[section]
\newtheorem{prop}[thm]{Proposition}

\newtheorem{cor}[thm]{Corollary}
\newtheorem{lem}[thm]{Lemma}

\theoremstyle{remark}
\newtheorem*{oss}{{\bf Remark}}

\begin{document}

\title{Hydrodynamic limit for a diffusive system \\ with boundary conditions}
\author{Stefano Marchesani}

\date{}
\maketitle

\abstract
 {
We study the hydrodynamic limit for the isothermal dynamics of an
anharmonic chain under hyperbolic space-time scaling and with nonvanishing viscosity. 
The temperature is kept constant by a contact with a heat bath, realised via a
stochastic momentum-preserving noise added to the dynamics. The
noise is designed so it  contributes to the macroscopic limit. Dirichlet boundary conditions are also considered:
one end of the  chain is kept fixed, while a time-varying tension
is applied to the other end. Moreover, Neumann boundary conditions are added in such a way that the system produces the correct thermodynamic entropy in the macroscopic limit.
We show that the volume stretch and
momentum converge (in an appropriate sense) to a smooth solution of a
system of parabolic conservation laws (isothermal Navier-Stokes equations in
Lagrangian coordinates) with boundary conditions. 

Finally, changing the external tension allows us to define
thermodynamic isothermal transformations between equilibrium states. We use this
to deduce the first and the second law of Thermodynamics for our
model.}
\bigbreak


{\let\thefootnote\relax
\footnote{{\today}}
}

\section{Introduction}
We consider a Hamiltonian system of anharmonic elastic springs. As it is well-known, Hamiltonian systems have poor ergodic properties. This means that, in general, there are conserved quantities other than total mass (or, in the case of a chain, length), momentum and energy, that ''survive'' as the number of particles goes to infinity. While it is expected that a suitable choice of the interaction among the particles might be able to provide ergodicity, this is still an important open problem. The typical solution is then to add a stochastic perturbation to the Hamiltonian dynamics that ''destroys'' all the extra conserved quantities and makes the dynamics ergodic.

Therefore, we put our system in contact with a heat bath which fixes the temperature along the chain. The heat bath is realised via a stochastic perturbation which acts both on the momenta (physical noise) and  the positions (artificial noise). Although fixing the temperature violates the conservation of energy, the noise is designed in such a way that it conserves the total length and momentum of the chain, at least away from the boundaries. Dirichlet boundary conditions are considered: one end of the chain is kept fixed, while the other end is pulled by a time-dependent external force. Finally, extra boundary conditions of the Neumann type are added: this is required so that the system produces the correct Clausius inequality in the macroscopic limit. 

This same system has been considered in \cite{MarchOlla1,MarchOllaHyperClausius} (without the extra Neumann conditions) and \cite{Fritz1} (for an infinite chain). A main feature is that the noise on the positions is nonlinear and has the same nonlinearity as the Hamiltonian interaction (see equation \eqref{eq:SDE} below). This kind of noise grants an easy bound on the Dirichlet form on the positions, which is heavily used in \cite{MarchOlla1, Fritz1}, as it leads to a crucial two-block estimate.

In the present paper we do not make use of a two-block estimate, as only the one-block estimate is enough to prove the hydrodynamic limit. In fact, the noise on the momenta alone is expected to provide such an estimate, as in \cite{even2010hydrodynamic}.

Nevertheless, a nonlinear noise translates into a nonlinear viscosity in the macroscopic limit. This gives rise to an interesting parabolic p-system with nonlinear viscosity and boundary conditions of mixed Dirichet-Neumann type. %The combination of these features thus makes for an interesting system to obtain, and it will be object of further investigations.

We obtain the hydrodynamic limit as a result of the \emph{relative entropy method}. This requires the existence of regular enough solutions to the macroscopic equations. However, although the macroscopic system is parabolic, it is also \emph{nonlinear}. Thus, we may not assume existence of global classical solutions, and such an existence needs to be proven. This is done in \cite{AlasioMarchesani}, where more general systems and boundary conditions are considered.

The relative entropy method for the full $3 \times 3$ Euler system has been employed in both \cite{olla1993hydrodynamical, even2010hydrodynamic}. 

In \cite{olla1993hydrodynamical} a tridimensional dynamics is studied in a bounded domain with periodic boundary conditions. Eulerian coordinates are used and because of that the classical (quadratic) kinetic energy yields a cubic term in the energy current. Such a term fails to be controlled by the relative entropy, and thus the kinetic energy is modified so it grows linearly at infinity (as an example one may think of the relativistic kinetic energy). In this way the energy current grows also linearly and thus can be controlled by the relative entropy.

In order to avoid modifying the kinetic energy, in \cite{even2010hydrodynamic} they work in Lagrangian coordinates, as no cubic current appears, in this setting. The paper also employs the same Dirichlet boundary conditions we are imposing, although a different noise is considered and no Neumann condition are added, as the macroscopic system in \cite{even2010hydrodynamic} is hyperbolic.

This article is structured as follows: Section 2 describes the dynamics and defines its invariant measures. In Section 3 we describe the macroscopic equations both in the variables $(r,p)$ and their conjugate $(\tau, p)$. Section 4 is devoted to the statement and proof of the hydrodynamic limit. Finally, in Section 5 we derive the first and the second law of Thermodynamics.

\section{Dynamics and Gibbs measures}
We study a system of $N+1$ particles coupled via anharmonic springs. The position of the $i$-th particle is $q_i \in \bb{R}$, and its momentum is $p_i \in \bb{R}$. The $0$-th particle is kept fixed at the origin, i.e. $(q_0,p_0) \equiv (0,0)$, while to the $N$-th particle is applied a time-dependent tension $\taubar(t)$. Particles $i$ and $i-1$ interact via a potential $V(q_i-q_{i-1})$ depending only on the relative position of nearest neighbours. The function $V:  \bb{R} \to \bb{R}_+$ is assumed to be smooth and strongly convex, meaning there exist strictly  constants $C_-$ and $C_+$ such that
\begin{equation}
0<C_- \le V''(r) \le C_+, \qquad \forall r \in \bb{R}.
\end{equation}
For ${\bf q} := (q_0, \dots, q_N)$ and $\pv := (p_0, \dots p_N)$, the energy for the system is defined through the Hamiltonian
\begin{equation}
\mathcal H_N ({\bf q}, \pv,t) :=\frac{p_0^2}{2}+ \somma{i}{1}{N} \left( \frac{p_i^2}{2}+V(q_i-q_{i-1}) \right) - \taubar(t) q_N.
\end{equation}
Since the interaction depends on the distances of nearest neighbours, we define the interparticle distances
\begin{equation}
r_i := q_i-q_{i-1}, \qquad 1 \le i \le N.
\end{equation}
Consequently, since we also have $p_0 = 0$, the phase space is given by $(\bb{R}^N)^2$ and the Hamiltonian takes the form $\mathcal H_N(\rv, \pv,t) = \somma{i}{1}{N} e_i$, where
\begin{equation}
e_i := \frac{p_i^2}{2}+V(r_i)-\taubar(t) r_i, \quad 1\le i \le N
\end{equation}
is the one-particle energy.  The system is then put in contact with a heat bath that acts as a microscopic stochastic viscosity.   

If we perform a hyperbolic space-time scaling, choose a non-vanishing viscosity and define the discrete gradient and Laplacian by
$$
\nabla a_i := a_{i+1}-a_i, \qquad \Delta a_i := a_{i+1}+a_{i-1}-2a_i
$$ 
the evolution equations are given by the following system of stochastic differential equations:
\begin{equation} \label{eq:SDE}
\begin{cases}
d r_1 =  Np_1 d t + \deltauno N^2 \nabla V'(r_1) d t - \sqrt{2 \beta^{-1}\deltauno} N \, d \widetilde w_1
\\
d r_i =  N\nabla p_{i-1} d t + \deltauno N^2 \Delta V'(r_i) d t - \sqrt{2 \beta^{-1} \deltauno} N\,
\nabla d\widetilde w_{i-1}, & 2 \le i \le N-1\\
d r_N =  N\nabla p_{N-1} dt + \deltauno N^2 \left(\taubar(t) + V'(r_{N-1}) - 2V'(r_N)\right)
- \sqrt{2 \beta^{-1} \deltauno} N\, \nabla d\widetilde w_{N-1},
\\
d p_1 = N \nabla V'(r_1) d t + \deltadue  N^2\left(p_2 - 2p_1\right) dt - \sqrt{2 \beta^{-1} \deltadue}N
\,  \nabla dw_0,\\
d p_j =  N\nabla V'(r_j) d t +\deltadue N^2\Delta p_j  d t - \sqrt{2 \beta^{-1} \deltadue}N
\,  \nabla dw_{j-1}, & 2 \le j \le N-1
\\
d p_N =N (\taubar(t)-V'(r_N)) d t -\deltadue  N^2\nabla p_{N-1}d t
+ \sqrt{2\beta^{-1}\deltadue }N\, d  w_{N-1}
\end{cases}
\end{equation}
Here $t \ge 0$ is the macroscopic time, $\beta^{-1} >0$ is the temperature and $\{\tilde w_i\}_{i=1}^{N}$, $\{ w_i\}_{i=0}^{N-1}$ are independent
families of independent Brownian motions. Note that we added Neumann boundary conditions for $r_1$ and $p_N$, as no Laplacian in the viscosity term appears there.

The boundary tension $\taubar : \mathbb R_+ \to \mathbb R$ is smooth and such that there exist a time $T_\star$ after which $\taubar$ is constant. Note that the boundary tension changes at a \emph{macroscopic} time scale.

The dynamics is generated by
\begin{equation}
\mathcal G^{\taubar(t)}_N:= N L_N^{\taubar(t)}+ N^2 \left(\delta_1 \Stilde_N^{\taubar(t)} + \delta_2 S_N \right).
\end{equation}
The Liouville operator $L_N^{\bar \tau(t)}$ is given by
\begin{equation}
L_N^{\bar \tau(t)}  = p_1 \frac{\partial}{\partial r_1}+ \sum_{i = 2}^N (p_i-p_{i-1})\frac{\partial}{\partial r_i}+
\sum_{i =1}^{N-1} \left(V'(r_{i+1}) -V'(r_i) \right)  \frac{\partial}{\partial p_i} +(\bar \tau(t) - V'(r_N)) \frac{\partial}{\partial p_N}.
\end{equation}

The operators $S_N$ and $\Stilde_N$ generate the stochastic part of the dynamics and are defined by
\begin{equation}
S_N := -\beta^{-1}\sum_{i=0}^{N-1}  D^*_i D_i, \quad \tilde S_N^{\taubar(t)} := - \beta^{-1}\sum_{i=1}^N  \Dtilde^*_i \Dtilde_i,
\end{equation}
where, for $1 \le i \le N-1$,
\begin{align}
\qquad D_i := \frac{\partial }{\partial p_{i+1}}-\frac{\partial }{\partial p_i},& \qquad D^*_i := \beta(p_{i+1}-p_i) -D_i
\\
\Dtilde_i  := \frac{\partial }{\partial r_{i+1}}-\frac{\partial }{\partial r_i},& \qquad \Dtilde^*_i :=\beta \left( V'(r_{i+1})- V'(r_i)\right) -  \tilde D_i.
\end{align}
The extra boundary operators are
\begin{align}
D_0 := \frac{\partial}{\partial p_1},& \qquad D_0^* := \beta p_1 - D_0,
\\
\Dtilde_N := -\frac{\partial}{\partial r_N},& \qquad \Dtilde_N^* := \beta(\taubar(t)-V'(r_N))-\Dtilde_N.
\end{align}
For $\bar p, \tau \in \mathbb R$ and $\beta >0$ we define the following family of Gibbs measures as
\begin{equation}
\nu_{\tau, \pbar, \beta}^N(d \rv, d\pv) := \prod_{i=1}^N \e{ \beta\tau r_i + \beta \pbar p_i- \beta\left(\frac{p_i^2}{2}+V(r_i)\right)- G(\tau,\bar p,\beta)}dr_idp_i,
\end{equation}
where $G$ is the Gibbs potential
\begin{align}
G(\tau,\pbar,\beta) : & = \log \int_{-\infty}^{+\infty}dr \int_{-\infty}^{+\infty}dp\, \e{ \beta\tau r + \beta \pbar p- \beta\left(\frac{p^2}{2}+V(r)\right)} \nonumber
\\
&= \sqrt\frac{2\pi}{\beta}+\beta \frac{\pbar^2}{2} +\log \int_{-\infty}^{+\infty} dr\, \e{\beta \tau r- \beta V(r)}.
\end{align}
We observe that, for constant $\taubar$, the Gibbs-measure $\nu^N_{\bar \taubar,0,\beta}$ is invariant for the dynamics generated by $\mathcal G_N^{\bar \taubar}$.

Setting, $\nu_{\tau,\pbar,\beta}:=\nu^1_{\tau,\pbar,\beta}$ we define the average elongation by 
\begin{equation}
\ell(\tau,\beta):= \attesa{r_1} = \frac{1}{\beta} \frac{\partial G}{\partial \tau}.
\end{equation}
It is standard to check (cf Appendix A of \cite{MarchOlla1}) that the function $\ell(\cdot, \beta)$ is strictly increasing and hence invertible. Its inverse shall be denoted by $\tau(\cdot, \beta)$.

Finally, we note that we have
\begin{equation}
\attesa{p_1} = \frac{1}{\beta}\frac{\partial G}{\partial \pbar} = \pbar \qquad \attesa{V'(r_1)} = \tau, \qquad   \attesa{p_1^2}  -\bar p^2= \beta^{-1},
\end{equation}
which identify $\bar p$ with the mean velocity, $\tau$ with the mean force (tension) and $\beta^{-1}$ with the temperature.
\section{The macroscopic equations}
Since the dynamics described in the previous section fixes the temperature to the constant value $\beta^{-1}$, we shall omit to write the dependencies on $\beta$. For example, we will simply write $\tau(r)$ instead of $\tau(r,\beta)$.

We expect that the empirical measures
\begin{equation}
r_N(t,dx) := \frac{1}{N}\somma{i}{1}{N}\delta\left(x-\frac{i}{N}\right) r_i(t)dx, \qquad p_N(t,dx) := \frac{1}{N}\somma{i}{1}{N}\delta\left(x-\frac{i}{N}\right) p_i(t)dx
\end{equation}
converge, in a suitable sense, to  absolutely continuous measures
\begin{equation}
r(t,x)dx, \quad p(t,x)dx
\end{equation}
whose densities $r$ and $p$ solve the parabolic system
\begin{equation} \label{system_pr}
\begin{cases} 
\partial_t r- \partial_x p= \delta_1 \partial_{xx}\tau(r)
\\
\partial_tp- \partial_x \tau(r)=\delta_2 \partial_{xx}p
\end{cases}
\end{equation}
with boundary conditions
\begin{equation}\label{eq:bc_pr}
p(t,0)=0, \quad r(t,1)=  \ell \left(\taubar(t) \right), \quad \partial_x p(t,1)=0, \quad \partial_x r(t,0)=0, \qquad \forall t \ge 0
\end{equation}
and initial data
\begin{equation}\label{eq:ic_pr}
p(0,x) = 0, \qquad r(0,x) = \ell(\taubar(0)), \qquad \forall x \in[0,1].
\end{equation}

Given our assumptions on $V$, the tension $\tau : \bb{R} \to \bb{R}$ is smooth, increasing and strictly convex. Furthermore, $\tau'$ is bounded away from zero (cf Appendix A of \cite{MarchOlla1}). 
In order to use the relative entropy method, we shall need equations for the Legendre conjugates (with respect to the Gibbs potential) of the variables $r(t,x)$ and $p(t,x)$. The conjugate of $r(t,x)$ is $\hat \tau(t,x):=\tau(r(t,x))$, while the conjugate of $p(t,x)$ is itself.

The function $\hat \tau(t,x)$ solves the equation
\begin{equation}
\partial_t \hat \tau(t,x)= \tau'(r(t,x))\partial_xp(t,x)+ \delta_1 \tau'(r(t,x))\partial_{xx}\hat \tau(t,x).
\end{equation}
with boundary conditions $\hat \tau(t,1)= \taubar(t)$ and $\partial_x\hat \tau(t,0)=0$. The latter follows from
\begin{equation}
\partial_x r(t,x) = \partial_x \ell(\hat \tau(t,x)) = \ell ' (\hat \tau(t,x)) \partial_x \hat \tau(t,x),
\end{equation}
$\partial_x r(t,0)=0$ and the fact that $ \ell'$ never vanishes ($\ell$ is strictly increasing).

Thus, provided identifying $r(t,x) = \ell(\hat \tau(t,x))$, the couple $(\hat \tau(t,x),p(t,x))$ solves the system
\begin{equation}\label{system_ptau}
\begin{cases}
\partial_t\hat \tau -\tau'(r)\partial_xp= \tau'(r)\delta_1 \partial_{xx}\hat \tau
\\
\partial_tp - \partial_x\hat \tau= \delta_2 \partial_{xx}p,
\end{cases},
\end{equation}
with boundary conditions
\begin{equation}
p(t,0)=0, \quad\hat \tau(t,1)= \taubar(t), \quad \partial_x p(t,1)=0, \quad \partial_x\hat \tau(t,0)=0
\end{equation}
and initial data
\begin{equation}
p(0,x) = 0, \qquad\hat \tau(0,x)= \taubar(0).
\end{equation}
Existence and uniqueness of global solutions of class $C^1$ in time and $C^2$ in space for systems \eqref{system_pr} and \eqref{system_ptau} with our initial-boundary conditions have been proven in \cite{AlasioMarchesani}. More precisely, we have
\begin{align}
r,p \in  C^1(\mathbb R_+ ; C^0([0,1]) )\cap C^0(\mathbb R_+; C^2([0,1])).
\end{align}
Therefore, in the following we shall assume $r(t,x)$ (or equivalently $\hat \tau(t,x)$) and $p(t,x)$ to be such regular solutions.

\section{Main theorem and relative entropy}
We define the local Gibbs measures as
\begin{equation}
g^N_t(\rv, \pv)d\rv d \pv:= \prod_{i=1}^N \e{\beta \TAU{i} r_i+ \beta \PI{i}p_i - \beta \left(\frac{p_i^2}{2}+V(r_i)\right) -G\left( \TAU{i},\PI{i},\beta\right)}d r_i dp_i
\end{equation}
and let $f^N_t(\rv, \pv)$ be the solution of the Fokker-Plank equation
\begin{equation}
\dfrac{\partial f_t^N}{\partial t} = \mathcal G^{\taubar(t), \dagger}_N f_t^N, \qquad f_0^N(\rv,\pv) = g_0^N(\rv,\pv),
\end{equation}
where $\dagger$ denotes the adjoint with respect to the Lebesgue measure on $\mathbb R^{2N}$.
If define the relative entropy as
\begin{equation}
H_N(t) := \int f_t^N \log \frac{f_t^N}{g_t^N}\, d\rv d\pv
\end{equation}
our aim is to prove the following
\begin{thm}[Main theorem] \thlabel{thm:main}
Denote by $d\mu_t^N = f_t^N(\rv, \pv) d\rv, d\pv$ the probability distribution of the system at time $t \ge 0$, starting from the local Gibbs measure $d \nu^N_0 :=g_0^N(\rv, \pv) d\rv d\pv$ corresponding the the initial profiles $r(0,x)$ and $p(0,x)$. Let ${\bf u}_i := (r_i, p_i)$ and ${\bf u}(t,x) :=(r(t,x),p(t,x))$. Then, for any continuous function $J : [0,1] \to \bb{R}$ and any $\varepsilon >0$,
\begin{equation}
\lim_{N\to \infty} \mu_t^N \left( \left| \frac{1}{N} \somma{i}{1}{N} J \left(\frac{i}{N}\right){\bf u}_i- \int_0^1 J(x){\bf u}(t,x)dx \right| > \varepsilon \right) = 0,
\end{equation}
where ${\bf u} \in C^1(\mathbb R_+ ; C^0([0,1]) )\cap C^0(\mathbb R_+; C^2([0,1]))$ is a solution of the system \eqref{system_pr} with boundary conditions \eqref{eq:bc_pr} and initial conditions \eqref{eq:ic_pr}.
\end{thm}
We shall prove the main theorem as a consequence of the following
\begin{thm} \thlabel{thm:entropy}
\begin{equation}
\lim_{N \to \infty} \frac{H_N(t)}{N} = 0
\end{equation}
for all $t \ge 0$.
\end{thm}
\begin{proof}[Proof of  \thref{thm:main}]
For any two probability measures $\alpha$, $\beta$, such that $\alpha$ is absolutely continuous with respect to $\beta$, define the relative entropy
\begin{equation}
H(\alpha|\beta) := \int \log \frac{d \alpha}{d \beta} d\alpha,
\end{equation}
where $d\alpha / d \beta$ is the Radon-Nikodym derivative of $\alpha$ with respect to $\beta$. Then, for any measurable $h$ and any $\sigma >0$, the following entropy inequality holds:
 \begin{equation}
 \int h d \alpha \le \frac{1}{\sigma} \log \int e^{\sigma h} d \beta + \frac{1}{\sigma} H(\alpha|\beta).
 \end{equation}
In particular, if $h = 1_A$ is the indicator function of the set $A$, we obtain 
\begin{align}
\begin{split}
\alpha(A) & = \int h d \alpha \le \frac{1}{\sigma}\log \int e^{\sigma 1_A}d\beta + \frac{1}{\sigma}H(\alpha|\beta)
\\
& = \frac{1}{\sigma} \log\left( \beta(A) (e^\sigma-1)+1\right)+\frac{1}{\sigma}H(\alpha|\beta).
\end{split}
\end{align}
Choosing $\sigma = \log \left(1+\dfrac{1}{\beta(A)}\right)$ then gives 
\begin{equation}\label{eq:entropy_A}
\alpha(A) \le \frac{\log 2+H(\alpha|\beta)}{\log \left(1+\dfrac{1}{\beta(A)}\right)}.
\end{equation}
Thus, if we define
\begin{equation}
A_\varepsilon := \left\{ \left| \frac{1}{N} \somma{i}{1}{N} J \left(\frac{i}{N}\right){\bf u}_i- \int_0^1 J(x){\bf u}(t,x)dx \right| > \varepsilon \right\},
\end{equation}
for any continuous $J: [0,1] \to \bb{R}$, then thanks to Theorem \ref{thm:entropy} and \eqref{eq:entropy_A}, to prove $\mu_t^N(A_\varepsilon) \to 0$ as $N \to \infty$, it is enough to show that, for each $\varepsilon >0$,
\begin{equation}
\log \left(1+ \frac{1}{\nu_t^N(A_\varepsilon)}\right) \ge C(\varepsilon)N,
\end{equation}
for some constant $C(\varepsilon)$ independent of $N$. However, this is satisfied if
\begin{equation}
\nu_t^N(A_\varepsilon) \le e^{-C(\varepsilon)N},
\end{equation}
which is a standard result of the large deviation theory \cite{kipnis2013scaling, varadhan1988large}.
\end{proof}
The following Lemma follows from Lemma 1.4 of Chapter 6 of \cite{kipnis2013scaling} after a time integration and using the fact that $H_N(0)=0$.
\begin{lem}
\begin{equation} \label{eq:entropy}
H_N(t) \le  \int_0^t ds \int \frac{f_s^N}{g_s^N}\left[\left(\mathcal G_N^{\taubar(s)}\right)^\dagger-\partial_s\right] g_s^N\, d\rv d\pv.
\end{equation}
\end{lem}
\begin{oss}
Recalling that $\dagger$ denotes the adjoint with respect to the Lebesgue measure, it is clear that  $\left(L_N^{\taubar(t)}\right)^\dagger=-L_N^{\taubar(t)}$. Moreover, for $1 \le i \le N-1$
\begin{align}
\begin{split}
\left(-\beta^{-1}D_i^* D_i\right)^\dagger & = \left[ -(p_{i+1}-p_i)D_i+\beta^{-1}D_i^2\right]^\dagger = \left(p_{i+1}-p_i+\beta^{-1}D_i\right) D_i + 2
\end{split}
\end{align}
and
\begin{align}
(-\beta^{-1}D_0^* D_0)^\dagger & = \left[ -p_1 D_0+ \beta^{-1}D_0^2\right]^\dagger = (p_1+\beta^{-1}D_0)D_0+1.
\end{align}
Therefore, we obtain
\begin{equation}
S_N^\dagger =2N-1+ \somma{i}{0}{N-1}D_i^\flat D_i,
\end{equation}
where
\begin{equation}
D_i^\flat := p_{i+1}-p_i+\beta^{-1}D_i, \qquad 1 \le i \le N-1
\end{equation}
and
\begin{equation}
D_0^\flat := p_1+ \beta^{-1}D_0.
\end{equation}
Similarly, we obtain
\begin{equation}
\left(\Stilde_N^{\taubar(t)} \right)^\dagger = \somma{i}{1}N \left[\Dtilde^\flat_i \Dtilde_i+V''(r_{i+1})+V''(r_i) \right], \qquad V''(r_{N+1})\equiv 0,
\end{equation}
where
\begin{equation}
D_i^\flat := V'(r_{i+1})-V'(r_i)+ \beta^{-1}\Dtilde_i, \qquad 1 \le i \le N-1
\end{equation}
and
\begin{equation}
\Dtilde^\flat_N := \taubar(t)-V'(r_N)+\beta^{-1} \Dtilde_N.
\end{equation}
\end{oss}
We will now evaluate the right-hand side of \eqref{eq:entropy}. In the following we shall denote by $a_N(t)$ a generic function such that
\begin{equation}
\lim_{N\to +\infty} \frac{1}{N} \int_0^t \int a_N(s) f_s^N\, d\pv d\rv ds =0.
\end{equation}
\begin{lem}[Liouville generator] 
\begin{equation} \label{eq:liouville}
\frac{N L_N^{\taubar(t)} g_t^N}{g_t^N}=\beta\somma{j}{1}{N} \left\{\partial_x \TAU{j}\left[ \PI{j}-p_{j-1}\right]+\partial_x  \PI{j} \left[\TAU{j}-V'(r_j)\right] \right\} + a_N(t).
\end{equation}
\end{lem}
\begin{proof}
By Lemma 2 of \cite{even2010hydrodynamic}, with $\lambda_1 = \beta \hat \tau$, $\lambda_2 = \beta p$ and $\lambda_3 = - \beta$ we obtain
\begin{align}
 \frac{N L_N^{\taubar(t)} g_t^N}{g_t^N}=- \beta\somma{j}{1}{N} \left[\partial_x \TAU{j}p_{j-1}+\partial_x  \PI{j} V'(r_j) \right] + \beta N p(1,t) \taubar(t)+ a_N(t).
\end{align}
	Then note that we can write
	\begin{equation}
	p(1,t)\taubar(t) = \int_0^1 \frac{\partial}{\partial x} \left[ \hat \tau(t,x)p(t,x)\right] dx= \int_0^1\left[\partial_x \hat \tau(t,x)p(t,x)+\partial_xp(t,x)\hat \tau(t,x)\right]dx.
	\end{equation}
	Thus we have
	\begin{equation}
	\beta N \taubar(t)p(1,t) = \beta \somma{j}{1}{N} \left[ \partial_x \TAU{j} \PI{j}+ \partial_x \PI{j} \TAU{j} \right] + a_N(t),
	\end{equation}
which completes the proof.
\end{proof}
In the same way, we obtain
\begin{lem}[Explicit time derivative]
\begin{align}
\begin{split}
\frac{\partial_t  g_t^N}{g_t^N}  =& \beta \somma{j}{1}{N} \left[ \tau' \left( \RHO{j}\right) \partial_x \PI{j}+ \delta_1 \tau' \left( \RHO{j} \right) \partial_{xx} \TAU{j} \right] \left[ r_j- \RHO{j}\right]+
\\
& + \beta \somma{j}{1}{N} \left[ \partial_x \TAU{j}+ \delta_2 \partial_{xx}\PI{j}\right]\left[p_j-\PI{j}\right].
\end{split}
\end{align}
\end{lem}
\begin{lem}[Physical viscosity]
\begin{equation}
\frac{N^2 S_N^\dagger  g_t^N}{g_t^N}   =\beta \somma{j}{1}{N-1}\partial_{xx}\PI{j}\left[ p_j-\PI{j}\right]+ a_N(t).
\end{equation}
\end{lem}
\begin{proof}
\begin{equation}
S_N^\dagger g_t^N =  \somma{j}{0}{N-1} D^\flat_j  \left \{g_t^ND_j\somma{i}{1}{N} \left[ \beta \PI{i}p_i-\beta \frac{p_i^2}{2}\right]\right \}+ (2N-1)g_t^N
\end{equation}
\begin{align}
=& \beta\somma{j}{1}{N-1} D_j^\flat \left\{ g_t^N \left[ \PI{j+1}-\PI{j}-(p_{j+1}-p_j)\right]\right\}+
\\
&+\beta D_0^\flat\left\{g_t^N \left[\PI{1}-p_1\right]\right\}+ (2N-1)g_t^N \nonumber
\end{align}
\begin{align} \label{eq:phys1}
=& \beta g_t^N \somma{j}{1}{N-1} \left[ \PI{j+1}-\PI{j} \right] (p_{j+1}-p_j)-\beta g_t^N \somma{j}{1}{N-1}(p_{j+1}-p_j)^2+
\\
&+ \beta g_t^N\somma{j}{1}{N-1} \left[ \PI{j+1}-\PI{j}-(p_{j+1}-p_j)\right]^2- (2 N-1)g_t^N + \nonumber
\\
&+\beta g_t^N \PI{1} \left[\PI{1}-p_1\right]+ (2 N-1)g_t^N. \nonumber
\end{align}
\begin{align}
=&- \beta g_t^N \somma{j}{1}{N-1} \left[ \PI{j+1}-\PI{j} \right] (p_{j+1}-p_j)+
\\
&\label{eq:phys4}+ \beta g_t^N \somma{j}{1}{N-1} \left[ \PI{j+1}-\PI{j}\right]^2+\beta g_t^N \PI{1} \left[\PI{1}-p_1\right]. \nonumber
\end{align}
After a summation by parts and using $\partial_xp(1,t)=0$, we obtain
\begin{align}
&-\beta  g_t^N\somma{j}{1}{N-1}\left[ \PI{j+1}-\PI{j} \right] (p_{j+1}-p_j) \nonumber
\\
=& \beta   g_t^N\somma{j}{1}{N-1}  \left[ \PI{j+1}+\PI{j-1}-2\PI{j}\right]p_j+
\\
&- \beta g_t^N  \left[ p(1,t)-\PI{N-1}\right]p_N+\beta  g_t^N\left[\PI{1}-p(0,t)\right]p_1 \nonumber
\\
\label{eq:phys2} =& \frac{\beta}{N^2}  g_t^N\somma{j}{1}{N-1} \partial_{xx}\PI{j}p_j + \beta  g_t^N \left[\PI{1}-p(0,t)\right]p_1 +\frac{ a_N(t)}{N^2}g_t^N.
\end{align}
Since $p(0,t) =0$, combining the boundary terms of  \eqref{eq:phys1} and \eqref{eq:phys2} gives
\begin{align}
\beta g_t^N \PI{1} \left[\PI{1}-p_1\right]+  \beta  g_t^N \left[\PI{1}-p(0,t)\right]p_1&= g_t^N\left[\PI{1}\right]^2  \nonumber
\\
&= \mathcal O \left( \frac{1}{N^2}\right)
\end{align}
and are thus negligible.

Finally, we write
\begin{align}
0&= \frac{\beta }{N} \partial_xp(0,t) p(0,t) \nonumber
\\
&= -\frac{\beta }{N} \int_0^1 \frac{\partial}{\partial x}\left[\partial_x p(x,t) p(x,t) \right]dx \nonumber
\\
&=-\frac{\beta }{N} \int_0^1 \partial_{xx}p(x,t) p(x,t)dx -\frac{\beta}{N}\int_0^1\left[\partial_x p(x,t) \right]^2 dx \nonumber
\\
& = -\frac{\beta }{N^2} \somma{j}{1}{N-1} \partial_{xx} \PI{j}\PI{j}-\beta  \somma{j}{1}{N-1} \left[ \PI{j+1}- \PI{j} \right]^2+\frac{a_N(t)}{N^2},
\end{align}
which yelds
\begin{align}
\label{eq:phys3}  \beta g_t^N \somma{j}{1}{N-1} \left[ \PI{j+1}- \PI{j} \right]^2= -\frac{\beta g_t^N}{N^2} \somma{j}{1}{N-1} \partial_{xx} \PI{j}\PI{j}+\frac{a_N(t)}{N^2}
\end{align}
Thus, combining \eqref{eq:phys4}, \eqref{eq:phys2} and \eqref{eq:phys3} and using the fact that the boundary terms are negligible lead to the conclusion.
\end{proof}
\begin{lem}[Artificial viscosity]
\begin{equation}
\frac{N^2\left(\Stilde_N^{\taubar(t)}\right)^\dagger g_t^N}{g_t^N}  = \beta \somma{j}{1}{N-1}\partial_{xx}\TAU{j}\left[ V'(r_j)-\TAU{j}\right]+  \beta N \partial_x\hat \tau(1,t)\left[\taubar(t)-V'(r_N)\right]+a_N(t)
\end{equation}
\end{lem}
\begin{proof}
By a calculation analogous to the one of the previous lemma, we have
\begin{align}
\left(\Stilde_N^{\taubar(t)}\right)^\dagger g_t^N =&- \beta g_t^N \somma{j}{1}{N-1} \left[ \TAU{j+1}-\TAU{j}\right]\left[V'(r_{j+1})-V'(r_j)\right] + 
\\
&+\beta g_t^N \somma{j}{1}{N-1}\left[\TAU{j+1}-\TAU{j}\right]^2+ \nonumber
\\
&+ \beta \Dtilde_N^\flat \left \{ g_t^N \Dtilde_N \left[\hat \tau(1,t)r_N-   V(r_N) \right]\right\} +  g_t^N V''(r_N) \nonumber
\end{align}
By a direct computation, and recalling that $\hat \tau(1,t)=\taubar(t)$,
\begin{equation}
\beta \Dtilde_N^\flat \left \{ g_t^N \Dtilde_N \left[\hat \tau(1,t)r_N-   V(r_N) \right]\right\} +  g_t^N V''(r_N) = 0.
\end{equation}
After a summation by parts, we obtain
\begin{align}
&- \beta \somma{j}{1}{N-1} \left[ \TAU{j+1}-\TAU{j}\right]\left[V'(r_{j+1})-V'(r_j)\right]  \nonumber
\\
=& \frac{\beta}{N^2} \somma{j}{1}{N-1} \partial_{xx}\TAU{j}V'(r_j)- \frac{\beta}{N} \partial_x\hat \tau(1,t) V'(r_N) + \frac{a_N(t)}{N^2}. 
\end{align}
The conclusion then follows after remembering that $\partial_x\hat \tau(0,t)=0$, adding and subtracting
\begin{align}
\frac{\beta}{N} \partial_x\hat \tau(1,t)\taubar(t)& = \frac{\beta}{N} \int_0^1\frac{\partial}{\partial x} \left[ \partial_x\hat \tau(x,t)\hat \tau(x,t)\right]dx \nonumber
\\
&= \frac{\beta}{N} \int_0^1 \partial_{xx}\hat \tau(t,x)\hat \tau(t,x)dx + \frac{\beta}{N} \int_0^1 \left[\partial_x\hat \tau(x,t)\right]^2dx
\end{align}
and replacing integrals by summations.
\end{proof}
We show that the error we make when replacing $\taubar(t)$ by $V'(r_N)$ is controlled by the relative entropy.
\begin{lem}
\begin{equation}
 \int_0^t \int  \left|\taubar(s)-V'(r_N)  \right| f_s^N d\rv d \pv ds   \le \frac{C}{N} \left( 1+t+ \int_0^t H_N(s)ds \right) + \frac{1}{2}\frac{H_N(t)}{N}
\end{equation}
for some $C >0$ independent of $N$.
\end{lem}
\begin{proof}
We can write
\begin{align}
\taubar(t)- V'(r_N) = \Stilde_N^{\taubar(t)} \somma{i}{1}{N}  r_i& = \frac{1}{\delta_1 N^2} \mathcal G_N^{\taubar(t)} \somma{i}{1}{N} r_i- \frac{1}{\delta_1 N} \mathcal L_N^{\taubar(t)} \somma{i}{1}{N} r_i 	\nonumber
\\
&= \frac{1}{\delta_1 N^2}\mathcal G_N^{\taubar(t)} q_N - \frac{1}{\delta_1 N}p_N.
\end{align}
This yields
\begin{align}
\int_0^t \int \left[ \taubar(s)-V'(r_N)\right] f_s^N  d\rv d\pv ds =& \frac{1}{\delta_1 N^2}\int q_Nf_t^N d\rv d\pv- \frac {1}{\delta_1N^2} \int q_Nf_0^N d\rv d\pv + \nonumber
\\
&+ \frac{1}{\delta_1 N} \int_0^t \int p_Nf_s^N d\rv d\pv ds,
\end{align}
The conclusion then follows as a standard application of the entropy inequality. In fact,
\begin{align}
	\frac{1}{N} \int |p_N| f_s^N d \rv d \pv & \le \frac{1}{N} \log \int e^{|p_N|} g_s^N d \rv d \pv+ \frac{H_N(s)}{N} \nonumber
	\\
	& \le \frac{C}{N} + \frac{H_N(s)}{N}
\end{align}
Furthermore,
\begin{align}
	\frac{1}{N^2}\int|q_N| f_t^N d \rv d\pv & \le \frac{1}{N^2}\somma{i}{1}{N}\int |r_i| f_s^N d \rv d \pv \nonumber
	\\ 
	& \le \frac{1}{4 N^2}\somma{i}{1}{N}\log \int e^{4|r_i|} g_s^N d \rv d\pv +\frac{1}{4} \frac{H_N(t)}{N} \nonumber
	\\
	& \le \frac{C}{N} +\frac{1}{4} \frac{H_N(t)}{N}.
\end{align}
\end{proof}
So far we have obtained
\begin{align}
\frac{1}{2}\frac{H_N(t)}{N} &\le\frac{1}{N} \somma{i}{1}{N-1} \int_0^t \int\partial_x\TAU{i}(p_{i-1}-p_i) f_s^N d\rv d \pv ds+ 
\\
&+\frac{1}{N}\somma{i}{1}{N-1}\int_0^t \int \left[\partial_x \PI{i}+\delta_1 \partial_{xx}\TAU{i}\right] \times \nonumber
\\
 &\times \left\{V'(r_i)-\TAU{i}-\tau' \left( \RHO{i}\right)\left[r_i-\RHO{i}\right] \right\} f_s^N d\rv d \pv ds \nonumber
\\
&+\frac{C}{N} \int_0^t H_N(s) ds+\int_0^t \int a_N(s)  f_s^N d\rv d \pv ds. \nonumber
\end{align}
By a summation by parts, it is easy to see that the term
\begin{equation}
\frac{1}{N} \somma{i}{1}{N-1} \int_0^t \int\partial_x\TAU{i}(p_{i-1}-p_i) f_s^N d\rv d \pv ds
\end{equation}
vanishes as $N \to \infty$, up to terms proportional to $\int_0^t H_N(s)/N ds$. Thus, we shall discard it from now on.

The next step is to pass to averages on blocks of size $k \ll N$.  This will allow us to replace $V'$ by $\tau$ in the sense of \thref{thm:1block}. In order to introduce such blocks, we  cut away the boundaries  by restricting to configurations $\{ [Nl], \dots, N-[Nl] \}$, for some small $l>0$ such that $l \to 0$ after $N \to \infty$ and $lN \gg k$. This is done using the inequality (cf Proposition 4.5 of \cite{MarchOllaHyperClausius}), 
\begin{equation}
\left |\frac{1}{N} \somma{i}{1}{N} J\left(\frac{i}{N}\right) \psi(r_i,p_i) - \frac{1}{N} \somma{i}{[Nl]}{N-[Nl]} J \left(\frac{i}{N} \right) \frac{1}{2k+1} \sum_{|j-i| \le k} \psi(r_j,p_j) \right | \le C \left(l+\frac k N \right)^{1/2} \left( \frac 1 N \sum_{i=1}^N(r_i^2+p_i^2) \right)^{1/2}
\end{equation}
which holds for any smooth $J : [0,1] \to \bb{R}$ and any linearly growing $\psi: \bb{R}^2 \to \bb{R}$.  Since a standard application of the entropy inequality (cf Proposition 3.2 of \cite{MarchOllaHyperClausius}) yields the energy estimate
\begin{align}
 \int \frac 1 N \sum_{i=1}^N(r_i^2+p_i^2) f_t^N d \rv d\pv \le C,
\end{align}
we obtain
\begin{align}
\frac{H_N(t)}{N} \le&  \frac{1}{N}\somma{i}{[Nl]}{N-[Nl]} \int_0^t \int \left[\partial_x \PI{i}+\delta_1 \partial_{xx}\TAU{i}\right] \times
\\
& \times \left\{ \bar V'_{k,i}-\TAU{i}-\tau' \left( \RHO{i}\right)\left[\bar r_{k,i}-\RHO{i}\right] \right\} f_s^N d\rv d \pv ds  \nonumber
\\
&+\frac{C}{N} \int_0^t H_N(s) ds+\int_0^t \int a_{N, k,l}(s)  f_s^N d\rv d \pv ds , \nonumber
\end{align}
where we have set
\begin{equation}
\bar V'_{k,i}:= \frac{1}{2k+1}\sum_{|j-i|\le k}V'(r_j), \qquad  \bar r_{k,i}:= \frac{1}{2k+1}\sum_{|j-i|\le k}r_j,
\end{equation}
and where
\begin{equation}
\lim_{ l \to 0}\lim_{k\to \infty} \lim_{N \to \infty}\int_0^t \int a_{N,k,l}(s)  f_s^N d\rv d \pv ds  =0.
\end{equation}
We replace $\bar V'_{k,i}$ by $\tau (\bar r_{k,i})$ via the one block estimate, which proof can be found in Proposition A.2 of \cite{MarchOllaHyperClausius}.
\begin{thm}[One-block estimate] \thlabel{thm:1block}
\begin{equation}
\lim_{l \to 0}\lim_{k \to \infty}\lim_{N \to \infty} \frac{1}{N} \somma{i}{[Nl]}{N-[Nl]} \int_0^t \int \left ( \bar V'_{k,i} - \tau(\bar r_{k,i}) \right)^2 f_s^N d \rv d \pv ds = 0.
\end{equation}
\end{thm}
\begin{oss}
Note that we did not need to cut unbounded variables, as in \cite{olla1993hydrodynamical, olla2014microscopic}, but we perform, in the fashion of \cite{Fritz1, MarchOlla1,MarchOllaHyperClausius} , an explicit estimate which makes use of the fact that $\tau$ is linearly bounded.
\end{oss}
Thus we have obtained
\begin{align}
\frac{H_N(t)}{N} \le &  \frac{1}{N}\somma{i}{[Nl]}{N-[Nl]} \int_0^t \int \left[\partial_x \PI{i}+\delta_1 \partial_{xx}\TAU{i}\right] \times
\\
& \times \left\{ \tau(\bar r_{k,i})-\TAU{i}-\tau' \left( \RHO{i}\right)\left[\bar r_{k,i}-\RHO{i}\right] \right\} f_s^N d\rv d \pv ds  \nonumber
\\
& +\frac{C}{N} \int_0^t H_N(s) ds+\int_0^t \int a_{N,k,l}(s)  f_s^N d\rv d \pv ds. \nonumber
\end{align}
Next, we  write
\begin{align}
\frac{H_N(t)}{N} \le & \frac{1}{N} \somma{i}{[Nl]}{N-[Nl]} \int_0^t \int  \Omega \left(t, \frac i N, \bar r_{k,i} \right) f_s^N d\rv d\pv ds +
\\
&+\frac{C}{N} \int_0^t H_N(s) ds+\int_0^t \int a_{N,l,k}(s)  f_s^N d\rv d \pv ds, \nonumber
\end{align}
where we have set
\begin{equation}
\Omega( t , x, \xi) :=\left[ \partial_x p(t,x)+ \delta_1 \partial_{xx} \hat \tau(t,x)  \right]\left\{ \tau(\xi) - \hat \tau(t,x) - \tau'(r(t,x)) [\xi - r(t,x)]\right\}
\end{equation}
Note that $\Omega(t,x, r(t,x) ) = \partial_\xi \Omega(t,x, r(t,x)) = 0$.
 
Consequently, Varadhan's lemma \cite{kipnis2013scaling, varadhan1988large} applies as in Theorem 4 of \cite{olla2014microscopic}, and we obtain
\begin{equation}
\frac{H_N(t)} N \le C \int_0^t \frac{H_N(s)} N ds + \int_0^t R_{N,k,l}(s)ds,
\end{equation}
for some uniform constant $C$, where
\begin{equation} \label{RNkl}
\lim_{l \to 0}\lim_{k \to \infty}\lim_{N \to \infty}  \int_0^t R_{N,k,l}(s)ds = 0.
\end{equation}
It then follows by Gronwall inequality that
\begin{equation}
\frac{H_N(t)} N \le \frac{H_N(0)} N e^{Ct} + \int_0^t R_{N,k,l}(s) e^{C(t-s)}ds \le \frac{H_N(0)} N + e^{Ct} \int_0^t R_{N,k,l}(s) ds .
\end{equation}
This gives
\begin{equation}
\lim_{N \to \infty} \frac{H_N(t)}{N} = 0,
\end{equation}
since $H_N(0) =0$ and the fact that $R_{N,k,l} \to 0$ in the sense of \eqref{RNkl}.

\section{Thermodynamic consequences}
This final section is devoted to the study of the Thermodynamics for the macroscopic system obtained as result of the hydrodynamic limit. Recall that the temperature is fixed from the dynamics to the constant value $\beta^{-1}$. Therefore, we shall consider \emph{isothermal transformations} between equilibria given by different values of the external tension $\taubar$.

We shall rigorously derive the second law of Thermodynamics in the form of the Clausius inequality. Moreover, upon assuming that the energy converges (which the hydrodynamic limit does \emph{not} allow us to do), we will obtain the first law, too.

Such a procedure can be found in \cite{olla2014microscopic} for an isothermal transformation in a case where the macroscopic equation is a single diffusive equation. The underlying hydrodynamic limit was obtained there with a diffusive space-time scaling.

An early result about the Clausius inequality for a diffusive system can be found in Appendix B of \cite{MarchOlla1}. However, such a result is purely macroscopic and does not follow from the hydrodynamic limit.

In \cite{MarchOllaNote} the Clausius inequality has been proven for vanishing viscosity solutions to the hyperbolic system obtained from our system by taking $\delta_1=\delta_2=0$. This is done entirely at the macroscopic level and takes into account the fact that shocks might arise as the viscosity vanishes.

Finally, in \cite{MarchOllaHyperClausius} the Clausius inequality is derived directly from the microscopic system we consider in this article with the same space-time scaling but with \emph{vanishing viscosity}. The macroscopic system is then hyperbolic and we allow the presence of shocks.
\subsection{The Clausius inequality}
In order to highlight the fact that we are performing an isothermal transformation, we shall restore the dependencies on $\beta$ throughout this section. Define the total free energy at time $t$ as
\begin{equation}
\mathcal F(t):= \int_0^1 \left[ \frac{p(x,t)^2}{2} +F(r(x,t),\beta)\right]dx,
\end{equation}
where
\begin{equation}
 F(r,\beta) =\int_0^r \tau(\xi,\beta) d\xi
\end{equation}
is the equilibrium free energy.
\begin{prop} \label{prop:preClausius}
For any $t \ge 0$ and $\delta_1,\delta_2>0$,
\begin{equation} \label{eq:clausiust}
\mathcal F(t)-\mathcal F(0) = \int_0^t\taubar(s) \mathcal L'(s)ds-  \int_0^t \int_0^1\delta_2  \left(\partial_x p\right)^2+\delta_1  \left(\partial_x \tau(r,\beta)\right)^2 ds dx,
\end{equation}
where
\begin{equation}
\mathcal L(t) := \int_0^1 r(t,x)dx
\end{equation}
is the total length of the chain at time $t$.
\end{prop}
\begin{proof}
Whenever there is an integral in both space and time, we shall omit to write the dependence of $r$ and $p$ on $x$ and $t$. Write
\begin{align}
\mathcal F(t)- \mathcal F(0) = & \int_0^t \frac{d}{ds} \mathcal F(s) ds=\int_0^t \int_0^1 p \partial_s p+\tau(r,\beta)\partial_s r\,ds dx
\\
=&\int_0^t \int_0^1p\partial_x \tau(r,\beta) +\tau(r,\beta)\partial_x p \, ds dx + \label{eq:Clau0}
\\
&+ \int_0^t \int_0^1\delta_2 p \partial_{xx} p+\delta_1 \tau(r,\beta) \partial_{xx}\tau(r,\beta)\, dsdx \nonumber
\end{align}
After an integration by parts in space, we have
\begin{align}
&\int_0^t \int_0^1p \partial_x \tau(r,\beta) +\tau(r,\beta)\partial_x p\, ds dx  \nonumber
\\
=& \int_0^t p(1,s)\tau(r(1,s),\beta)-p(0,s) \tau(r(0,s),\beta)ds\nonumber
\\
 =& \int_0^t \taubar(s) p(1,s)ds\nonumber
 \\
= & \int_0^t \taubar(s) \int_0^1 \partial_x p \, dx ds \nonumber
\\
=& \int_0^t\taubar(s)\int_0^1 \left[\partial_s  r dx -\delta_1   \partial_{xx} \tau(r,\beta)  \right] dx ds, \nonumber
\\
=& \int_0^t\taubar(s) \mathcal L'(s)ds-\delta_1 \int_0^t \taubar(s) \partial_x \tau(r(1,s),\beta)ds. \label{eq:Clau1}
\end{align}
where
\begin{equation}
\mathcal L(s):=  \int_0^1 r(x,s)dx.
\end{equation}
Finally, using the Neumann boundary conditions $\partial_x p(1,t)=\partial_x r(0,t)=0$ we obtain
\begin{align}
&\int_0^t \int_0^1p \delta_2 \partial_{xx} p+\delta_1 \tau(r,\beta) \partial_{xx}\tau(r,\beta)dsdx \nonumber
\\
=&-  \int_0^t \int_0^1\delta_2  \left(\partial_xp\right)^2+\delta_1  \left(\partial_x \tau(r,\beta)\right)^2 ds dx + \delta_1 \int_0^t \taubar(s) \partial_x  \tau \left(r(1,s),\beta\right) ds. \label{eq:Clau2}
\end{align}
When we use \eqref{eq:Clau1} and \eqref{eq:Clau2} in \eqref{eq:Clau0}, the boundary terms cancel exactly, and we get the conclusion.
\end{proof}
In order to obtain the Clausius inequality from the previous lemma we shall define an isothermal thermodynamic transformation as follows. Recall that the system at time zero is at equilibrium with tension $\bar \tau(0) := \tau_0 \in \mathbb R$ and temperature $\beta^{-1}$ namely
\begin{align}
p(0,x) = 0 \qquad \tau(r(0,x),\beta) = \tau_0 \qquad \forall x \in [0,1].
\end{align}
In particular, we have
\begin{align}
\mathcal F(0) = F(\ell(\tau_0),\beta).
\end{align}
Now we take $\bar \tau$ to vary smoothly from $\tau_0$ to $\tau_1 \in \mathbb R$ in a finite time $T_\star$ and to stay at the value $\tau_1$ for all subsequent times. Then, after waiting a long time, the system reaches a new equilibrium at tension $\tau_1$ and temperature $\beta^{-1}$, in the sense of the following
\begin{prop} \label{prop:asymp}
\begin{equation}
 \lim_{t \to \infty} \mathcal F(t) = F(\ell(\tau_1),\beta).
\end{equation}
\end{prop}
\begin{proof}
Define $\mathcal F_{\taubar(t)}(t) = \mathcal F(t)- \taubar(t) \mathcal L(t) + \hat G(\taubar(t))$, where $\hat G$ is the Legendre transform of $F(\cdot, \beta)$. Thanks to Proposition \ref{prop:preClausius}, recalling that $\hat G'=\ell$ and that $\tau'$ is positive and bounded away from zero we compute
\begin{align}
\frac{d}{dt}\mathcal F_{\taubar(t)} (t) &= -\taubar'(t) \mathcal L(t) + \ell(\taubar(t))\taubar'(t) - \int_0^1 \delta_2 (\partial_x p(t,x))^2 +\delta_1(\partial_x \tau(r(t,x), \beta))^2 dx \nonumber
\\
& \le  -\taubar'(t)  \left[\mathcal L(t)-\ell(\taubar(t))\right]  - C \int_0^1 (\partial_x p(t,x))^2 +(\partial_x r(t,x))^2 dx \nonumber
\\
&=-\taubar'(t)  \left[\mathcal L(t)-\ell(\taubar(t))\right]  - C\int_0^1 (\partial_x p(t,x))^2 + \left\{\partial_x[ r(t,x) -\ell(\taubar(t)) ]\right\}^2 dx.
\end{align}
Since $p(t, x)$ vanishes at $ x =0$ and $r(t,x)-\ell(\taubar(t))$ vanishes at $x=1$, we can apply Poincar\`e inequality in order to obtain
\begin{align} \label{eq:med}
\frac{d}{dt}\mathcal F_{\taubar(t)} (t)  \le- \taubar'(t)  \left[\mathcal L(t)-\ell(\taubar(t))\right]   - C\int_0^1  p(t,x)^2 +  \left[ r(t,x)-\ell(\taubar(t)) \right]^2 dx.
\end{align}
Observe that $F(r, \beta)-\taubar(t)r+\hat G(\taubar(t))$ is a uniformly convex function of $r$ and vanishes, together with its first derivative, if $r = \ell(\taubar(t))$. Hence, we may find  some constants $C_2 > C_1 > 0$ such that
\begin{align}
C_1 [r-\ell(\taubar(t))]^2 \le F(r,\beta)- \taubar(t) r + \hat G(\taubar(t)) \le C_2 [r-\ell(\taubar(t))]^2
\end{align}
and we may estimate the integral at the right hand side of \ref{eq:med} by $- C \mathcal F_{\taubar(t)}(t)$.

Furthermore, take $t > T_\star$, where $T_\star$ is such that $\taubar(t) =\tau_1$ on $[T_\star,+\infty)$. Then $\taubar'(t)=0$ and we obtain
\begin{align}
\frac{d}{dt} \mathcal F_{\tau_1}(t) \le -C \mathcal F_{\tau_1}(t), \qquad \forall t  > T_\star.
\end{align}
Thus, Gronwall's inequality yields
\begin{align}
\mathcal F_{\tau_1}(t)  \le \mathcal F_{\tau_1}(T_\star) e^{-C(t-T_\star)}, \qquad \forall t > T_\star
\end{align}
so that 
\begin{align} \label{eq:quasiaa}
0 = \lim_{t \to \infty} \mathcal F_{\tau_1}(t) & = \lim_{t\to\infty} \mathcal F(t) -\tau_1\lim_{t\to\infty}  \int_0^1 r(t,x)dx + \hat G(\tau_1) 
\\
&=\lim_{t \to\infty} \int_0^1 \frac{p(t,x)^2}{2} dx + \lim_{t\to \infty} \int_0^1 F(r(t,x),\beta)dx-\tau_1 \int_0^1r(t,x)dx + \hat G(\tau_1) . \nonumber
\end{align}
Since $F(r,\beta)-\tau_1 r + \hat G(\tau_1)$ is convex and non-negative, by Jensen's inequality we obtain
\begin{align}
 F \left(\lim_{t \to \infty} \int_0^1 r(t,x)dx,\beta \right) - \tau_1 \lim_{t\to \infty}\int_0^1 r(t,x) dx + \hat G(\tau_1) = 0.
\end{align}
This in turn implies
\begin{align}
\lim_{t\to\infty} \int_0^1 r(t,x) dx = \ell(\tau_1).
\end{align}
Finally, plugging this last relation into the first line of \eqref{eq:quasiaa} leads to the conclusion.
\end{proof}
Combining Propositions \ref{prop:preClausius} and \ref{prop:asymp} gives the following
\begin{thm}[Clausius inequality]
\begin{equation}
 F(\tau_1, \beta) -  F(\tau_0,\beta) \le  W,
\end{equation}
where
\begin{equation}
 F(\tau_i ,  \beta) := F\left( \ell(\tau_i ), \beta\right), \quad i =0,1
\end{equation}
is the the equilibrium free energy as function of  tension and temperature, and
\begin{equation} \label{eq:macrowork}
W := \int_0^\infty \taubar(s) \mathcal L'(s)ds
\end{equation}
is the total work done by the external force $\bar \tau$ during the transformation.
\end{thm}

\subsection{The first law of Thermodynamics}
The internal energy $U$ is defined as
\begin{equation}
U(\tau, \beta) :=  \left \langle \frac{p_1^2} 2 +V(r_1) \right\rangle_{\tau,0,\beta}.
\end{equation}
Define the microscopic average energy at time $t$ as
\begin{equation} 
\mathcal E_N (t):= \frac{1}{N}\somma{i}{1}{N} \left(\frac{p_i^2(t)}{2}+V(r_i(t))\right).
\end{equation}
The law of large numbers for the initial distribution gives
\begin{equation} \label{eq:en0}
\lim_{N \to \infty} \mathcal E_N(0) = U(\tau_0,\beta)
\end{equation}
in probability. By the hydrodynamic limit and the convergence to equilibrium  we expect that
\begin{equation} \label{eq:ent}
\lim_{t \to \infty} \lim_{ N \to \infty} \mathcal E_N(t)  = U( \tau_1, \beta),
\end{equation}
 but at the present time we do not have the tools to prove it. This would require some knowledge about the finiteness of  expectations of powers of $p$ higher than the second, and the relative entropy (to date the main tool used in order to obtain microscopic estimates)  does not allow to control functions which grow more than the energy itself. Thus, we shall \emph{assume} that \eqref{eq:ent} holds.

Using the microscopic dynamics and omitting to write the dependences of $r_i$ and $p_i$ on time gives
\begin{align}
\mathcal E_N(t)-\mathcal E_N(0) =& \int_0^t \somma{i}{1}{N-1} p_i(V'(r_{i+1})-V'(r_i))ds +\int_0^t p_N (\taubar(s)-V'(r_N))ds+  
\\
&+  \int_0^t\somma{i}{1}{N-1} V'(r_i)(p_i-p_{i-1})ds + \int_0^t V'(r_N)(p_N-p_{N-1}) ds+ \nonumber
\\
& + N \int_0^t\left( \delta_1 \Stilde_N\somma{i}{1}{N} V(r_i) + \delta_2 S_N\somma{i}{1}{N} \frac{p_i^2}{2}\right)ds+ \nonumber
\\
& +\sqrt{2 \beta^{-1}\delta_2} \int_0^t \left(p_1 (d w_0-d w_1) + \somma{i}{2}{N-1}p_i(dw_{i-1}-dw_i)+p_N dw_{N-1} \right)+ \nonumber
\\
& +\sqrt{2 \beta^{-1}\delta_1} \int_0^t \left(-V'(r_1) d\tilde w_1 + \somma{i}{2}{N-1}V'(r_i)(d\tilde w_{i-1}-d\tilde w_i)+V'(r_N) (d\tilde w_{N-1}-d\tilde w_N) \right). \nonumber
\end{align}
\begin{align}
 =& \int_0^t \left[\taubar(s) p_N+ N\delta_1 V'(r_N)(\taubar(s)-V'(r_N))\right] ds+  
\\
&+N \delta_2 \int_0^t \left[ \beta^{-1}(2N-1)- \somma{i}{0}{N-1}(p_{j+1}-p_j)^2 ds \right]ds+ \nonumber
\\
& +N\delta_1 \int_0^t\left[\beta^{-1}V''(r_N)+ \beta^{-1}\somma{i}{1}{N-1} [V''(r_{i+1})+V''(r_i)] -\somma{i}{1}{N-1}[V'(r_{i+1})-V'(r_i)]^2ds\right]+  \nonumber
\\
&+\sqrt{2 \beta^{-1}\delta_2}\int_0^t p_1 dw_0 +\sqrt{2 \beta^{-1}\delta_2}\int_0^t\somma{i}{1}{N-1} (p_{i+1}-p_i)dw_i+  \nonumber
\\
&+\sqrt{2 \beta^{-1}\delta_1}\int_0^t\somma{i}{1}{N-1} (V'(r_{i+1})-V'(r_i))d\wtilde_i-\sqrt{2 \beta^{-1}\delta_1}\int_0^t V'(r_N) d\tilde w_N. \nonumber
\end{align}
We write
\begin{align}
\taubar(s) p_N+ N\delta_1 V'(r_N)(\taubar(s)-V'(r_N)) = & \taubar(s)[p_N+ N \delta_1(\taubar(s)-V'(r_N))] 
\\
& - N \delta_1[\taubar(s)-V'(r_N)]^2 \nonumber
\\
=& \taubar(s)\left[ d\left( \frac{1}{N} \somma{i}{1}{N}r_i \right) +\sqrt{2\beta^{-1} \delta_1} d\tilde w_N\right]- N \delta_1[\taubar(s)-V'(r_N)]^2 \nonumber
\\
=& \taubar(s) d \mathcal L_N(s)+ \taubar(s)\sqrt{2\beta^{-1} \delta_1} d\tilde w_N - N \delta_1[\taubar(s)-V'(r_N)]^2, \nonumber
\end{align}
where
\begin{align}
\mathcal L_N(s) :=  \frac{1}{N} \sum_{i=1}^N r_i(s). 
\end{align}
If we define the microscopic heat
\begin{align}
Q_N(t) := & \delta_2  \beta^{-1}N(2N-1) t- N \delta_2\int_0^t\somma{i}{0}{N-1}(p_{j+1}-p_j)^2  ds+
\\
& +N\delta_1 \beta^{-1} \int_0^t\left[V''(r_N)+ \somma{i}{1}{N-1} (V''(r_{i+1})+V''(r_i))\right] ds \nonumber
\\
& - N \delta_1 \int_0^1 \left[ (\taubar(s)-V'(r_N))^2+\somma{i}{1}{N-1}(V'(r_{i+1})-V'(r_i))^2 \right] ds+ \nonumber
\\
& + \sqrt{2 \beta^{-1}\delta_2} \int_0^t p_1 d w_0 +\sqrt{2 \beta^{-1}\delta_2}\int_0^t\somma{i}{1}{N-1} (p_{i+1}-p_i)dw_i+
\\
&+ \sqrt{2 \beta^{-1}\delta_1}\int_0^t\somma{i}{1}{N-1} (V'(r_{i+1})-V'(r_i))d\wtilde_i + \sqrt{2\beta^{-1} \delta_1} \int_0^t\left[\taubar(s) - V'(r_N)\right]d\tilde w_N \nonumber
\end{align}
and the microscopic work
\begin{align}
W_N(t) :&=  \int_0^t \taubar(s) d \mathcal L_N(s)
\end{align}
we obtain the microscopic version of the first law of thermodynamics:
\begin{equation} \label{eq:1lawN}
\mathcal E_N(t)-\mathcal E_N(0) = Q_N(t)+ W_N(t).
\end{equation}
Thanks to the hydrodynamic limit we can prove the following
\begin{prop} \label{prop:w}
\begin{align}
\lim_{N \to \infty} W_N(t) = \int_0^t \taubar(s) \mathcal L'(s)ds
\end{align}
in probability.
\end{prop}
\begin{proof}
Since $\taubar$ is deterministic, an integration by parts in time gives
\begin{align} \label{eq:wint}
W_N(t) = - \int_0^t \taubar'(s) \mathcal L_N(s) ds + \taubar(t) \mathcal L_N(t)- \taubar(0) \mathcal L_N(0)
\end{align}
Then, we apply Theorem \ref{thm:main} with $J = 1$ in order we obtain
\begin{align}
\mathcal L_N(s) = \frac{1}{N} \sum_{i=1}^N r_i(s) \to \int_0^1 r(s,x)dx =: \mathcal L(s) \qquad \forall s \ge 0
\end{align}
in probability. Therefore, taking the limit $N \to \infty$ in \eqref{eq:wint} and integrating by parts yields
\begin{align}
\lim_{N \to \infty}W_N(t) &=  - \int_0^t \taubar'(s) \mathcal L(s) ds + \taubar(t) \mathcal L(t) - \taubar(0) \mathcal L(0)
\\
& = \int_0^t \taubar(s) \mathcal L'(s) ds. \nonumber
\end{align}
\end{proof}
Applying \eqref{eq:en0},  \eqref{eq:ent} and Proposition \ref{prop:w} to \eqref{eq:1lawN}, we obtain that  $Q_N(t)$ converges, as $N \to \infty$ and $t\to\infty$, to the deterministic
\begin{equation} \label{eq:1stt}
Q:= U(\tau_1, \beta)-U(\tau_0,\beta) - W,
\end{equation}
where
\begin{equation}
W := \int_0^\infty\taubar(s) \mathcal L'(s)ds
\end{equation}
is the total work done by the external tension.  Thus, we have obtained the following
\begin{thm}[First law of thermodynamics]
\begin{align}
U(\tau_1, \beta)-U(\tau_0,\beta)= Q + W,
\end{align}
where 
\begin{align}
Q = \lim_{t \to \infty} \lim_{N \to \infty}Q_N(t)
\end{align}
is the total heat exchanged with the thermostats and 
\begin{align}
W = \lim_{t \to \infty} \lim_{N \to \infty} W_N(t)
\end{align}
is the total work done by the external tension.
\end{thm}
The Clausius inequality, together with the first law of thermodynamics, allow us to obtain the following
\begin{cor}[Second law of thermodynamics]
Let the thermodynamic entropy $S$ be defined as
\begin{align}
S(\tau,\beta) := \beta[ U(\tau, \beta) - F(\tau,\beta)].
\end{align}
Then,
\begin{align}
S(\tau_1,\beta) - S(\tau_0,\beta) \ge \beta Q.
\end{align}
\end{cor}
%\begin{proof}
%\begin{align}
%S(\tau_1,\beta)-S(\tau_0,\beta) &= \beta [ U(\tau_1,\beta)-U(\tau_0,\beta)] - \beta [ F(\tau_1,\beta)- F(\tau_0,\beta)] \nonumber
%\\
%&  \ge \beta(Q+W) - \beta W \nonumber
%\\
%& = \beta Q.
%\end{align}
%\end{proof}

\noindent
{Stefano Marchesani\\
GSSI, \\
{\footnotesize Viale F. Crispi 7, 67100 L'Aquila, Italy}}\\
{\footnotesize \tt stefano.marchesani@gssi.it}\\

\end{document}